\documentclass[12pt]{amsart}

\usepackage{macros}

\pdfoutput = 1

\begin{document}
\onehalfspacing

\title{\large Eleven-dimensional supergravity \\ as a Calabi--Yau twofold}

\author{Fabian Hahner$^\flat$ \and Ingmar Saberi$^\natural$}

\email{fhahner@mathi.uni-heidelberg.de}

\email{i.saberi@physik.uni-muenchen.de}

\address{{$^\flat$}Institut f\"ur Mathematik, Universit\"at Heidelberg \\ Im Neuenheimer Feld 205 \\ 69120 Heidelberg, Deutschland}

\address{{$^\natural$}Ludwig-Maximilians-Universit\"at M\"unchen \\ Theresienstra\ss{}e 37 \\ 80333 M\"unchen, Deutschland}

\begin{abstract}
	We construct a generalization of Poisson--Chern--Simons theory, defined on any supermanifold equipped with an appropriate filtration of the tangent bundle. Our construction recovers interacting eleven-dimensional supergravity in Cederwall's formulation, as well as all possible twists of the theory, and does so in a uniform and geometric fashion. Among other things, this proves that Costello's description of the maximal twist
	\emph{is} the twist of eleven-dimensional supergravity in its pure spinor description. 
	It also provides a pure spinor lift of the interactions in the minimally twisted theory. Our techniques enhance the BV formulation of the interactions of each theory to a homotopy Poisson structure by defining a compatible graded-commutative product; this suggests interpretations in terms of deformations of geometric structures on superspace, and provides some concrete evidence for a first-quantized origin of the theories.
\end{abstract}

\maketitle
\thispagestyle{empty}

\tableofcontents
\setcounter{tocdepth}{1}
\newpage

\setlength{\parskip}{6pt}

\section{Introduction}

Since the first supersymmetric field theories were constructed, it has been a goal to understand their properties and simplify their construction using superspace techniques. This motivation has perhaps been largest in the case of supergravity theories. The geometric nature of the theory of Einstein gravity, which is constructed using a covariant least-action principle on the space of metrics of Lorentzian signature, has motivated much research that tries to give an equally pithy formulation of supergravity theories as governing moduli problems of (deformations of) particular natural geometric structures on superspace. (The set of appropriate references is dense in the literature, so we will refrain.)

Among all supergravity theories of physical interest, perhaps the most exceptional is eleven-dimensional supergravity, which was first constructed by Cremmer, Julia, and Scherk in 1978~\cite{CJS}, and which is expected to be the low energy limit of M-theory~\cite{WittenM}. M-theory has yet to be constructed, although expectations exist that a worldsheet construction as a theory of fundamental membranes might be possible. While the component-field formulation of this theory is relatively streamlined---in addition to the metric, the theory contains only a gravitino and an abelian three-form gauge field with Chern--Simons term---it proved difficult to even formulate the theory in superspace, and a superspace least action principle remained out of reach. Part of the difficulty can be attributed to attempts to find sets of auxiliary fields that could be used to represent supersymmetry off shell, which was seen as a necessary prerequisite.

A major leap forward was taken in work of Cederwall, who applied the pure spinor superfield formalism to construct a superspace description of perturbative eleven-dimensional supergravity using the BV formalism. 
The relation of eleven-dimensional pure spinors to supergravity dates back at least to~\cite{HowePS2}.
The connection had been sharpened in~\cite{CederwallM5}, which observed that a particular eleven-dimensional pure spinor superfield reproduced the BV supergravity multiplet. In~\cite{Ced-towards}, a candidate cubic interaction term for this multiplet was constructed; in~\cite{Ced-11d}, Cederwall went on to extend this by a somewhat subtle quartic term in the BV action functional, and to prove that the result satisfies the BV master equation, thus giving a consistent, manifestly supersymmetric interacting theory that---since the theory is expected to be unique---must be eleven-dimensional supergravity itself. 
(Pure spinor techniques were also used from a first-quantized perspective to give new models of the supermembrane; see Berkovits' work in~\cite{BerkovitsSupermembrane}, generalizing his formulation of the superstring.)
The pure spinor description thus not only formulates the theory on superspace, but also dramatically simplifies the structure of  its interactions: a non-polynomial action for the component fields is replaced by a quartic polynomial. (Such simplifications are typical of interacting pure spinor theories.)
Nonetheless, it does not provide a geometric origin for the quartic polynomial in question. Neither does it give an interpretation of the moduli problem it describes in terms of deformations of the superspace geometry itself.

Later, and in disjoint fashion, further progress was made on twisted versions of eleven-dimensional supergravity. Twists of supergravity theories were defined by Costello and Li in~\cite{CostelloLi}, generalizing the standard notion of a twist of a supersymmetric field theory. Using worldsheet techniques from topological string theory, they gave a proposed description of the holomorphic twist of type IIB supergravity. Costello and Li's theory is a version of BCOV theory~\cite{BCOV}, for which the moduli-theoretic interpretation is clear; it is related to the Kodaira--Spencer theory of deformations of Calabi--Yau structure.  In~\cite{CostelloMtheory2}, Costello went on to investigate eleven-dimensional supergravity in the omega background; his proposed description links the maximal twist of eleven-dimensional supergravity to a non-commutative Chern--Simons theory called \emph{Poisson--Chern--Simons theory}.

Poisson--Chern--Simons theory is simple to describe in the BV formalism. Its fields are given by the Dolbeault complex of $(0,\bu)$-forms on a Calabi--Yau twofold, tensored with the de Rham complex on $\R^7$ (or, more generally, a $G_2$-manifold; for nonperturbative issues related to $G_2$-manifolds, see~\cite{TopoM,DOZG2} and references therein). The interactions are determined by an $L_\infty$ structure on the fields, which is in fact strict: the Lie bracket is the Poisson bracket of holomorphic functions induced by the Calabi--Yau form, whose inverse is a holomorphic Poisson bivector. This theory has two essential features. Firstly, it also has a moduli-theoretic interpretation. The Lie algebra of holomorphic functions with the Poisson bracket is a one-dimensional central extension of holomorphic Hamiltonian vector fields. Since the symplectic structure is the holomorphic volume form, these are also divergence-free vector fields, and can thus be also thought of as related to the moduli space of deformations of Calabi--Yau structures.
Secondly, the central extension equips the fields of Poisson--Chern--Simons theory with a \emph{commutative} structure; the interactions define not just a dg Lie structure, but a dg Poisson algebra structure. Recalling that the observables of a three-dimensional TQFT are equipped with an $E_3$-algebra structure, which is equivalent to an even-shifted Poisson structure, we see that this formulation is at least suggestive of a first-quantized origin. (Note, though, that there are subtleties in defining $E_3$ algebra structures on theories of this type; see~\cite{ChrisBrian}.)

Recent work has pushed our understanding of twisted eleven-dimensional supergravity further; all approaches have either used dualities or target-space techniques, since no worldsheet description is available. Pure spinor techniques were applied in~\cite{spinortwist} to give concise and computationally straightforward descriptions of the twists of supergravity multiplets. This led to the first direct computations of the minimally twisted eleven-dimensional and type IIB supergravity multiplets, the latter confirming Costello and Li's proposal at the free level. Working directly with the component fields, it was also shown that the maximally twisted eleven-dimensional multiplet reduces to Poisson--Chern--Simons theory in the free limit~\cite{MaxTwist}. In~\cite{RSW11d}, a consistent interacting $\Z/2$-graded BV theory was defined on the minimally twisted eleven-dimensional supergravity multiplet. Surprisingly, the cohomology of this theory on flat space is a one-dimensional $L_\infty$ central extension of the exceptional infinite-dimensional simple super Lie algebra $E(5|10)$~\cite{Kac}. Other exceptional simple super Lie algebras also play fundamental roles in holomorphic M-theory~\cite{RWindex,WS-e36}.

In this paper, we take a step towards bringing some of these lines of work together by exploiting a powerful and seemingly underappreciated analogy between the geometric structures in play on each case. Thinking of Poisson--Chern--Simons theory (after localizing six directions with omega backgrounds) as a theory in five dimensions, we note that the theory must be equipped with a \emph{transversely holomorphic foliation} (or THF structure) that lets us think of the geometry as locally isomorphic to $\C^2 \times \R$. The THF structure is an (involutive) three-dimensional subbundle of the complexified tangent bundle. Similarly, the minimally twisted theory is most generally defined on eleven-dimensional manifolds equipped with a six-dimensional complex distribution.

Flat superspace itself is also canonically equipped with a distribution, spanned by the left–invariant odd vector fields. However, since all bosonic translations are in the image of brackets of supersymmetry transformations, this distribution is as far from being integrable as possible. It is thus not possible to naively draw a connection between these two structures. A clue to the resolution is provided by the theory of Dolbeault cohomology for almost complex manifolds, recently developed in~\cite{ACDolbeault}. This theory uses the distribution $T^{(0,1)}$ to define a filtration of the de Rham complex. The differential on the associated graded measures the nonintegrability of the distribution; passing to its cohomology and transferring the $D_\infty$ structure defined by the remaining terms in the de Rham differential provides a new filtered complex, which they use as a replacement for the Hodge filtration. Passing to the associated graded of this new filtration defines their analogue of the Dolbeault complex.

If we apply the same construction to the de Rham complex on superspace, we can identify the term in the differential encoding the nonintegrability of the odd distribution with the Chevalley--Eilenberg differential of the supertranslation algebra. The ``generalized Dolbeault complex'' that appears is nothing other than the sum of the pure spinor multiplets associated to the cohomology groups of the supertranslation algebra; the cohomology in degree $-k$ plays the role of the Dolbeault complex resolving holomorphic $(k,0)$-forms. In particular, the canonical supermultiplet of~\cite{MSJI} appears playing the role of the holomorphic functions, and we think of it---equipped with its commutative structure---as the appropriate structure sheaf with which to equip the spacetime. In eleven dimensions, this is eleven-dimensional supergravity.

The analogy with complex geometry allows one to find ready generalizations of many interesting notions: the complex dimension is the degree of the highest Lie algebra cohomology of the supertranslations; a Calabi--Yau structure is a trivialization of (the multiplet of) top cohomology as a module over the structure sheaf. In this analogy, eleven-dimensional supergravity, and all of its twists, are Calabi--Yau twofolds. We use this to construct a family of theories we call \emph{homotopy Poisson--Chern--Simons theories}. The construction uses the derived bracket technique of~\cite{Kosmann96}, as generalized by~\cite{VoronovDerived}, and is entirely analogous to the standard construction of the Poisson bracket. However, because we work in a derived setting, the corresponding $L_\infty$ structure is in general not strict. Applying our construction recovers Cederwall's quartic interaction functional in geometric fashion, as well as Costello's maximal twist. Furthermore, it gives a pure spinor lift of the interactions of the minimal twist. It then follows from the results of~\cite{spinortwist}, which state that the twist of a canonical multiplet is the canonical multiplet of the twisted supersymmetry algebra, that these theories are all related by twisting, proving Costello's conjecture on the maximal twist at the full interacting level.

\subsection*{Structural overview} Here is a brief sketch of the structure of the paper. We begin in~\S\ref{sec:two} by recalling how a filtration of the tangent bundle by subbundles gives rise to a filtration of the sheaf of de Rham forms with certain additional properties. We abstract these properties into a notion of (weighted) flag structure on a cdga, and study an abstract version of Cirici and Wilson's generalization of the Fr\"olicher spectral sequence. This lets us define sheaves $W^{-k,\bu}$ of ``Dolbeault-resolved holomorphic $k$-forms'' on any (super)manifold equipped with such a structure. \S\ref{sec:three} reviews the construction of Poisson--Chern--Simons theory on products of Calabi--Yau twofolds and odd-dimensional real manifolds, as used in Costello's descriptions of maximally twisted eleven-dimensional supergravity, and then constructs a generalization to a ``homotopy'' version, defined for any appropriate weighted flag structure. \S\ref{sec:four} begins with a quick review of the pure spinor formalism; we rediscover this formalism here in terms of generalized Dolbeault complexes, but the reader who is unfamiliar with the standard story should start here for a few basic definitions. We then proceed to make some comments on its application to twisted theories and to characterize examples that give rise to ``Calabi--Yau twofolds.'' Finally, \S\ref{sec: 11d} constructs eleven-dimensional supergravity and its twists as examples of homotopy Poisson--Chern--Simons theories.

\subsection*{Acknowledgments}
We would like to give special thanks to: M.~Cederwall and C.~Elliott, for invaluable collaboration at an early stage of this project related to a pure spinor description of (holomorphic eleven-dimensional) supergravity; B. R. Williams and S. Raghavendran, for fruitful ongoing conversations as this work was taking shape, and for past, present, and future collaboration on related projects; J.~Huerta, for numerous insights and conversations, and in particular for his suggestion that Tanaka prolongation should be relevant. We also gratefully acknowledge conversations and collaborations on related topics with I.~Brunner, K.~Costello, R.~Eager, O. Gwilliam, S.~Jonsson, S. Noja, J.~Palmkvist, N.~Paquette, and J. Walcher, and offer special thanks to the anonymous referee for a careful reading of the draft and for many valuable suggestions.
This research was supported in part by Perimeter Institute for Theoretical Physics. Research at Perimeter Institute is supported by the Government of Canada through the Department of Innovation, Science and Economic Development
and by the Province of Ontario through the Ministry of Research, Innovation and Science. FH and IAS both thank the Perimeter Institute for its hospitality, and also M.~Bek, A.~Kiefer, and L.~Steinert for theirs.
IAS also thanks the Mainz Institute for Theoretical Physics of the DFG Cluster of Excellence PRISMA$^+$ (Project ID 39083149) for its hospitality during the workshop on Higher Structures, Gravity, and Fields, where part of this work was performed.
This work is funded by the Deutsche Forschungsgemeinschaft (DFG, German Research Foundation) under Germany’s Excellence Strategy EXC 2181/1 — 390900948 (the Heidelberg STRUCTURES Excellence Cluster) and Projektnummer 517493862 (Homologische Algebra der Supersymmetrie: Lokalität, Unitarität, Dualität), and by the Free State of Bavaria.

\section{Flag structures and generalized Dolbeault complexes}
\label{sec:two}

\numpar[p:conv][Conventions]
Throughout, we work in a category of complex super vector spaces equipped with an action of a semisimple Lie group of the form $\C^\times \times G$. Here, $G$ will depend on context, may be trivial, and will often be left implicit. The action of $\C^\times$ is equivalent to a grading by the integers. We refer to this grading as \emph{weight}, to emphasize that it is not a cohomological grading. The monoidal structure on this category is the one on super vector spaces; signs are thus determined by the $\Z/2$ grading and are independent of the weight. 

We will also consider cochain complexes of objects in this category. These are then equipped with two integer gradings (by cohomological degree and weight), as well as a $\Z/2$-grading by intrinsic parity. The Koszul sign is determined by the sum of cohomological degree and intrinsic parity, modulo two. Our conventions are always cohomological.

\subsection{Weighted flag structures}
We begin with some very general considerations, related to the type of geometric intuition we will draw on in the sequel. The essential point is to notice that certain (super or graded generalizations of) filtered structures, as studied by Tanaka, are present both on (almost) complex manifolds and on the superspaces of interest in physics. The resulting analogy between superspaces and almost complex manifolds will let us construct a sheaf of commutative differential graded algebras on such a manifold, which reproduces Dolbeault cohomology for complex manifolds, as well as its generalization to almost complex manifolds as defined in~\cite{ACDolbeault}. When we apply our techniques to superspaces, the construction naturally reproduces a particular multiplet 
in the pure spinor formalism.\footnote{The unfamiliar reader is referred forward to~\S\ref{sec: pure spinors} for terminology.} This is the multiplet assigned to the structure sheaf of the nilpotence variety, termed the \emph{canonical multiplet} in~\cite{MSJI} (and the \emph{tautological filtered cdgsa} in~\cite{spinortwist}). 

\numpar[p:flagdefs]
Geometrically, we will be interested in manifolds (including supermanifolds or graded manifolds) that are equipped with distributions. 
The definitions we give here are generalizations to the graded setting of standard definitions in the theory of \emph{Tanaka structures}~\cite{Tanaka}.\footnote{We owe deep thanks to John Huerta for calling our attention to the relevance of Tanaka's work.}
    We will not delve more deeply into connections to the theory of Tanaka prolongation, or to parabolic geometry more broadly, here, though these are certainly of great interest. We will return to them in future work; for now, the interested reader is referred to~\cite{Tanaka,Zelenko,AlekDavid,CapSlovak}.

Recall~\cite[chapter 19]{Lee} that a \emph{distribution} on a manifold $M$ is a subbundle $D \subset TM$ of the tangent bundle. A distribution is said to be \emph{involutive} if the space of vector fields lying in $D$ is a subalgebra of vector fields on $M$ with respect to the Lie bracket. A distribution is \emph{integrable} if there exists a regular foliation of $M$ such that the tangent spaces to the leaves agree with $D$ at each point. By Frobenius' theorem~\cite[Theorem 19.12]{Lee}, involutive distributions and integrable distributions coincide, and we will use the terms interchangeably. 

More generally, we can consider a flag of distributions, which is defined to be a finite sequence 
\deq[eq:flag]{
    0 \subset D_1 \subset \cdots \subset D_k = TM
}
of subbundles of the tangent bundle, each contained in the next. 
In other words, $D_\bu$ is a filtration of the tangent bundle.
Our conventions for filtrations follow~\cite[\S1]{DeligneHodge}, with the exception that we will write all filtrations as increasing.

We further require that $D_\bu$ is compatible with the Lie bracket of vector fields, in the sense that
\deq[eq:compat]{
    \left[ \Gamma(D_i), \Gamma(D_j) \right] \subset \Gamma(D_{i+j}).
}
Thus the sections of~$D_\bu$ give $\Vect(M)$ the structure of a filtered Lie algebra.\footnote{We emphasize that a filtered Lie algebra is \emph{not} filtered by sub Lie algebras: it is a filtered vector space with the compatibility~\eqref{eq:compat} between the filtration and the Lie bracket.}
When $k=2$, this condition is vacuous, so that we need only specify a single arbitrary distribution $D_1$. All of our examples here will be of this type.

The associated graded vector bundle $\Gr TM$, with $\Gr_p TM = D_p/D_{p-1}$, acquires the structure of a weight-graded Lie algebra in the category of vector bundles over~$M$. For $X \in \Gamma(D_i)$ and $Y \in \Gamma(D_j)$, one observes that 
\deq{
[X,fY] = X(f)Y + (-)^{|f|\cdot|X|} f[X,Y].
}
The first term is a section of~$D_j$, and is thus equivalent to zero in $\Gamma(\Gr_{i+j} TM)$. It follows that the Lie bracket on the associated graded is linear over functions. The fiber of this bundle of Lie algebras at $p\in M$ is called the \emph{symbol algebra} at~$p$; it is a finite-dimensional, strictly positively graded real Lie algebra. We will call the filtration of the tangent bundle \emph{regular} if the symbol algebras are isomorphic at every point of~$M$. 

\numpar
A filtration of the complexified tangent bundle determines a dual filtration of the complexified cotangent bundle~\cite[\S1.1.6]{DeligneHodge}, according to the rule  
\deq[eq:dualfilt]{
F_{n} T_\C^* = (T_\C /D_{-1-n})^\vee.
}
This convention ensures that $\Gr_p T_\C^* = (\Gr_{-p} T_\C )^\vee$.
We can extend this multiplicatively to a nonpositive filtration $F_\bu \Omega^\bu(M)$ of the complex de Rham forms. Since $D_\bu$ is compatible with the Lie bracket,  $F_\bu$ is preserved by the de Rham differential. 

As an example, consider the flag of distributions on an almost complex manifold defined by taking 
\deq[eq:HodgeFlag]{
   D_0 = 0 \subset  D_1 = T^{(0,1)} \subset D_2 = T_\C .
}
Applying~\eqref{eq:dualfilt}, we see that the dual filtration takes the form
\deq{
F_{-3} T^* = 0  \subset F_{-2} T^* = (T_\C/T^{(0,1)})^\vee \subset  F_{-1} T^*M = T_\C^*.
}
%As is clear from~\eqref{eq:ac-decomp} below, this filtration 
It is compatible with the de Rham differential for any almost complex structure. This filtration is called the ``shifted Hodge filtration'' $\tilde{F}$ in~\cite[Definition 3.5]{ACDolbeault}.

%Note that this filtration, although it is compatible with the de Rham differential, is \emph{not} the standard Hodge filtration. 
%Nor is it particularly convenient in applications. To recover the Hodge filtration, one needs to construct a new filtration $F^+_\bu \Omega^\bu(X)$, defined by taking
%\deq[eq:defF_+]{
%F_+^i \Omega^\bu(X) = \bigoplus_{k+j = i} F^j \Omega^k(X).
%}
%In the example of an almost-complex manifold, we then have that
%\deq{
%F_+^{-i} \Omega^\bu(X) = \Omega^{\geq i,\bu}(X).
%}
%When the complex structure is not integrable, the de Rham differential does not preserve $F^+_\bu$; see~\S\ref{sec:Dinfty} below.

\numpar[p:weighting][Compatible weight gradings]
Matters are simplified when the filtration of the tangent bundle arises from a grading. A compatible weight grading is an isomorphism
\deq{
\phi: \Gr TM \to TM
}
of filtered vector bundles, where $\Gr TM$ is filtered by weight. In other words,
%we have a decomposition of the tangent bundle via a positive integer grading that induces the flag of distributions we are interested in. We will refer to such a grading as a \emph{weight grading}. It consists of a direct sum decomposition of the tangent space of the form
%\deq{
%TM = \bigoplus_{1\leq j \leq k} T_j M,
%}
%such that the flag of distributions we are interested in is recovered by taking
\deq{
    \phi \left( \bigoplus_{1 \leq j \leq k} \Gr_j TM  \right) =  D_k.
}
We will sometimes write $T_j M$ for the summand $\phi(\Gr_j TM)$.
For both almost complex manifolds and superspaces, there is a canonical choice of such a splitting: in the first case by the eigenspaces of $J$, and in the second by the tangent space of the body of~$M$.

A compatible weight grading defines a compatible weight grading (in strictly negative degrees) on the cotangent bundle, and thus also on the de Rham complex of~$M$. 
In the example of an almost complex manifold, the grading assigns weight $-1$ to $\d \bar{z}$ and weight $-2$ to $\d {z}$. 

We draw a diagram of the weight grading in this example in Figure~\ref{fig:hodge}. The cohomological degree is on the vertical axis, and the weight is on the horizontal axis. The filtration $F_\bu$ is the column filtration: $F_{-k}$ consists of all summands with weight $\leq -k$. The dashed arrow represents the Nijenhuis tensor, and the dotted arrow its conjugate; see~\eqref{eq:ac-decomp} below. We remark that the standard Hodge filtration, which is not compatible with the differential in the almost-complex case, appears diagonally in the diagram.

\begin{figure}
\begin{equation*}
\begin{tikzcd}[column sep = 3 ex]
\text{weight:} & 0 & -1 & -2 & -3 & -4 & -5 & -6 \\
3) & & & &  \Omega^{0,3} & \Omega^{1,2} & \Omega^{2,1} & \Omega^{3,0} \\
2) & & &  \Omega^{0,2} \ar[ur,"\dbar"] \ar[urr] \ar[dotted,urrr] & \Omega^{1,1}  \ar[ur] \ar[urr] \ar[dashed, u] \ar[dotted,urrr] & \Omega^{2,0}  \ar[ur]\ar[urr] \ar[dashed, u]  \\
1) & & \Omega^{0,1} \ar[ur,"\dbar"] \ar[urr] \ar[dotted,urrr] & \ar[urr] \Omega^{1,0}\ar[ur] \ar[urr] \ar[dashed, u] \\
0) & \Omega^{0,0}\ar[ur,"\dbar"] \ar[urr] 
\end{tikzcd}
\end{equation*}
\caption{The weight grading arising from an almost-complex structure}
\label{fig:hodge}
\end{figure}

\numpar[p:flagstructure] Motivated by the previous considerations, we now give definitions which are meant to abstractly model the structures that are present on the de Rham complex of a manifold equipped with a flag of distributions (and perhaps with a compatible weight grading).
\begin{dfn}
Let $(\Omega^\bu,\d)$ be a cdga concentrated in nonnegative homological degree. A \emph{flag structure} on~$\Omega^\bu$ is an (ascending) filtration $F_\bu \Omega^k$ of each homogeneous summand, compatible both with the differential and the product, such that $F_0 \Omega^\bu = \Omega^\bu$, $F_{-1} \Omega^\bu = \Omega^{\geq 1}$, and $F_\bu \Omega^k$ is a filtration of finite length in each homological degree $k$. 
    A  \emph{weighted flag structure} on $\Omega^\bu$ consists of a weight grading in non-positive degrees, 
    %with respect to which the differential decomposes into pieces of non-positive weight. In other words,
    such that the column filtration associated to the weight grading is a flag structure.
\end{dfn}

From our perspective, there are two essential and natural examples of flag structures. The first of these, as we have already seen, is the ``shifted'' Hodge filtration on the de Rham complex of an (almost) complex manifold, as in Figure~\ref{fig:hodge}. The second is related to the examples in supersymmetric field theory that we have in mind as applications: any flat superspace is equipped with a canonical distribution, defined by considering the span of all translation-invariant odd vector fields. More generally, the supermanifolds that are valid backgrounds for supersymmetric field theories or supergravity theories are equipped with a regular non-involutive odd distribution of maximal dimension, modelling the local supersymmetry transformations. (This is well-known; consider, for example, the definition of a super Riemann surface~\cite{Friedan,RSV,WittenSuperRiemann}. The idea goes back at least to Manin in~\cite{Manin1,Manin2}, where such a datum is thought of as a \emph{superconformal structure}.)
    
In some sense, the usefulness of the definition lies in the fact that it brings these two examples under one roof. In particular, our main application---to eleven-dimensional supergravity---will rely on exploiting the analogy between the two. To get to these examples, we need to construct the generalization of Dolbeault cohomology to this more general setting. This will be done in the next section. We then move on to discuss examples in~\S\ref{sec: examples}.

\subsection{$D_\infty$ algebras from weighted flag structures}
\label{sec:Dinfty}
\numpar 
Given a weighted flag structure, 
%we can regrade $\Omega^\bu$ with respect to the sum of the weight grading and the cohomological grading. (The filtration associated to this totalized grading recovers $F^+_\bu$.) 
the differential $\d$ decomposes as a sum of terms 
\deq{
\d = \d_1 + \d_0 + \d_{-1} + \cdots,
}
where $\d_i$ has bidegree $(1,i-1)$. 
The subscript labels the \emph{totalized grading}, which is the sum of the cohomological degree and the weight..

The differential on $\Gr F_\bu(\Omega^\bu)$ can be identified with $\d_1$, which is a differential of square zero and bidegree $(1,0)$ on~$\Omega^\bu$.
We will now choose to regard this differential as ``internal,'' and the additional terms $\d_0 + \d_{-1} + \cdots$ as defining a further structure on $\Gr F_\bu(\Omega^\bu)$.

Recall that a square-zero endomorphism of degree one can be thought of as the defining data of an algebra structure over the operad $D$ governing square-zero differentials. (See, for example, \cite{Vallette}.) This operad has a single operation $\d_0$ of arity one and degree one, subject to the relation that its concatenation with itself vanishes.
%We view it as a dg operad in totalized degree zero. (From the perspective of the $D_\infty$ structure, the cohomological degree is the \emph{totalized} degree.)

A $D$-algebra structure on a cochain complex $(V,\d_1)$ is given by a single square-zero endomorphism $\d_0$ of degree one. Since $\d_0$ is a cochain map, $\d_0$ and $\d_1$ anticommute. Thus a cochain complex with a $D$-algebra structure is almost the same thing as a bicomplex, except for the fact that the second grading has been forgotten. To restore it, we specify an action of $\C^\times$ on the operad $D$ with respect to which the nontrivial operation has weight one, and ask for an equivariant $D$-algebra structure on a weighted cochain complex.

Due to the relation $\d_0^2 = 0$, the operad $D$ is not free, and does not play well with quasi-isomorphisms. As is standard in homotopical algebra, we must replace $D$ by a freely generated (weighted) dg operad that resolves it. This operad $D_\infty$ is generated by operations $\d_i$ for each nonpositive $i$, all of which have arity one, cohomological degree one, and weight $i-1$. The conditions defining a $D_\infty$ algebra structure in cochain complexes amount to the condition that the total differential 
\deq{
    \d = \d_1 + \d_0 + \d_{-1} + \cdots
}
is of square zero. Here, $\d_1$ again denotes the internal (weight-zero) differential of the cochain complex and $\d_{i}$ for $i \leq 0$ encode the $D_\infty$ algebra structure.
It is clear that a weighted flag structure defines a $D_\infty$ algebra structure on $\Gr F_\bu(\Omega^\bu)$.

\numpar[sec:HT][Homotopy transfer]
Since $D_\infty$ is a good homotopy replacement for $D$, one can use homotopy transfer of $D_\infty$ algebra structures to pass between different quasi-isomorphic models.\footnote{Our presentation is ahistorical; this technique, better known as the ``homological perturbation lemma,'' dates back to~\cite{Brown} and was probably the first example of homotopy transfer.}
%This encodes, in particular, the higher differentials of spectral sequence of a bicomplex.
We consider a model of $\Omega^\bu$ defined on the $E_1$ page of the spectral sequence associated to~$F_\bu$:
\deq{
    W^\bu \defeq H^\bu\left( \Gr F_\bu(\Omega^\bu)\right) = H^\bu\left(\Omega^\bu, \d_1\right).
}
Since $\d_1$ is homogeneous of bidegree $(1,0)$, $W^\bu$ is again bigraded by cohomological degree and by weight. %---or equivalently, by cohomological degree and totalized degree. We will find it more convenient to work with the totalized degree in the sequel.
We can apply the homotopy transfer theorem for $D_\infty$ algebras~\cite{operadsBook} in order to obtain a new $D_\infty$ algebra structure on $W^\bu$. In concrete terms, this is done by fixing a retraction
\begin{equation}
\begin{tikzcd}
\arrow[loop left]{l}{h}(\Omega^{\bu} \: , \: \d_1)\arrow[r, shift left, "p"] &(W^{\bu}  , \, 0)\arrow[l, shift left, "i"] \: .
\end{tikzcd}
\label{eq: hotop data}
\end{equation}
Although the transfer data depends on the choice of $i,p$ and $h$, we note that the transferred $D_\infty$ structure is unique up to isomorphism~\cite[Theorem 10.3.15]{operadsBook}.

The output of the above construction is a cdga $W^\bu$ with zero internal differential, equipped with a weighted $D_\infty$ structure. 
Alternatively, it is a weighted flag structure on $W^\bu$ with the property that the differential on $\Gr F_\bu(W^\bu)$ vanishes.
We will denote the terms of the $D_\infty$ structure by $\d'_i$ for $i \leq 0$; the term $\d'_i$ has cohomological degree one and weight $i-1$, and is thus of totalized degree $i$.  $\d' = \sum \d'_i$ is a square-zero differential of cohomological degree one and strictly positive weight. %, which now \emph{does} respect the filtration $F^+_\bu W^\bu$ associated to the totalized degree.

\numpar[sec:dec][Reweighting]
Since the differential on $W^\bu$ contains only terms of strictly positive weight, we are free to modify the $\C^\times$ action on~$W^\bu$, shearing it by the cohomological degree. The sheared weight acts on summands of cohomological degree $i$ and weight $j$ with weight $i+j$; it is thus just the totalized degree. Since $\d'$ contained only terms of weight $\leq -1$, it is of nonpositive sheared weight.

We note that the same procedure goes through for flag structures, without the choice of a weighting. We are guaranteed that the differential on $W^\bu$ decreases the filtration, in the sense that $\d'(F_k W^\bu) \subseteq F_{k -1}W^\bu$. Thus the sheared filtration $F^+_\bu$,  
defined by 
\deq{
F^+_k(W^j) \defeq F_{k-j}(W^j),
}
makes $(W^\bu, \d')$ into a filtered cdga that is quasi-isomorphic to~$(\Omega^\bu, \d)$.

The sheared filtration does \emph{not} define a flag structure on~$W^\bu$, since the condition that $\Gr_0 F^+_\bu W^\bu = W^0$ does not hold. (In complex geometry, this is just saying that the Dolbeault complex of $(0,\bu)$-forms is no longer supported in homological degree zero.) Nevertheless, $\Gr_0 F^+_\bu W^\bu$ will become the central character in our story.

\numpar[sec:A][The structure sheaf $A^\bu$; geometric interpretation]
%When applied to a weighted flag structure, 
%If we like, we can therefore repeat the procedure from above. $\Gr F^+_\bu W^\bu$ will be a bigraded cdga with a differential of totalized degree zero. If we were to shift the totalized grading up by the cohomological degree \emph{again}, we would get a $D_\infty$ structure on $\Gr F^+_\bu W^\bu$ with respect to that new grading. However, we will not have cause to do this. Instead, w
We have seen above that, for an integrable complex structure, $F^+_\bu$ is nothing other than the Hodge filtration. As was worked out in~\cite{ACDolbeault}, $F^+_\bu W^\bu$ is the correct object to replace the standard Hodge filtration (and thus the standard Dolbeault cohomology) for non-integrable complex structures. 
Building on this philosophy, we will regard $\Gr F^+_\bu W^\bu = (W^\bu, \d_0')$ as the fundamental object associated to a weighted flag structure. 

For weighted flag structures arising from regular filtrations of the tangent bundle, \S\ref{p:flagdefs} guarantees that $W^\bu$ arises as the smooth sections of a graded vector bundle on the underlying manifold, so that it is a cochain complex of locally free sheaves on the underlying manifold. 
We will allow ourselves to refer to it as the \emph{generalized Dolbeault complex}.
%We will discuss this in detail in examples in the next section.

The complex geometry of a complex manifold is governed by its sheaf of holomorphic functions; a good derived replacement for this sheaf is the sheaf $\Omega^{0,\bu}$ of Dolbeault forms that smoothly resolves it. There is an obvious generalization of this structure sheaf in our setting as well: 
it is just the sheaf of cdgas $A^\bu \defeq \Gr_0 F^+_\bu W^\bu$. (In Figure~\ref{fig:hodge} above, $A^\bu$ appears along the main diagonal in totalized degree zero, after passing to the cokernel of $\d_1$, depicted by dashed arrows.)
%. $W^\bu$ is negatively graded with respect to the totalized grading, so that we can decompose it as a sum
%\deq{
%W^\bu = \bigoplus_{i\leq 0} W^{i,\bu}
%}
%of homogeneous subspaces. This splitting is compatible with the differential on $\Gr F^+_\bu W^\bu$.
%As such, we can consider the cdga $A^\bu := (W^{0,\bu},\d'_0)$ sitting in totalized degree zero; t

Our central point is that \emph{any supermanifold with a regular filtration of the tangent bundle can (and should) be equipped with the structure sheaf $A^\bu$}. As we will see in the next section, applying this construction to examples arising from superspaces produces the canonical supermultiplet---and therefore, among other physically important examples, the eleven-dimensional supergravity multiplet. Pursuing this analogy with complex geometry further will allow us to produce the interactions of eleven-dimensional supergravity from a holomorphic Poisson structure on this ringed space, reproducing and generalizing work of Cederwall~\cite{Ced-towards,Ced-11d}.

\subsection{Examples of weighted flag structures} \label{sec: examples}
\numpar[sec:complex][Complex manifolds]

Let $X$ be a complex manifold; locally, we can equip $X$ with corresponding coordinates $(z^i,\bar{z}^i)$. We consider the de Rham complex on $X$,
\begin{equation}
    \left( \Omega^{\bullet} (X) \: , \: \d = \partial + \bar{\partial} \right) .
\end{equation}
The de Rham differential $\d$ splits into holomorphic and antiholomorphic pieces, the operators $\del$ and $\dbar$.
The cohomological grading is by form degree; to define the weight grading, we assign $\d \bar{z}$ weight zero and $\d z$ weight $-1$. This corresponds to the filtration 
\deq{
   0 \subset  D_1 = T^{(0,1)}X \subset D_2 = T_\C X
   }
   of the complexified tangent bundle, which we have refined to give a weighted flag structure by choosing
   \deq{
       T_1 X = T^{(0,1)} X, \quad T_2 X = T^{(1,0)} X.
   }
   In this example, it is clear that the terms of the decomposition of the differential are
   \deq{
       \d_1 = 0, \quad
       \d_0 = \dbar, \quad
       \d_{-1} = \del,
   }
   with all higher terms vanishing. As a result, $W^\bu$ can be identified with $\Omega^\bu$, and $A^\bu$ is the Dolbeault complex $\Omega^{0,\bu}(X)$.

\numpar[sec:almost][Almost complex manifolds]
Nothing in the construction of the weighted flag structure above depended on the integrability of the complex structure. In fact, the construction generalizes immediately to almost complex manifolds, with the difference that $D_1$ is no longer involutive. 
Correspondingly, the internal differential $\d_1$ no longer vanishes. We recover the theory of Dolbeault cohomology for almost complex manifolds, as worked out in~\cite{ACDolbeault}.

On an almost complex manifold, the de Rham differential decomposes as
	\begin{equation}
            \d = \Bar{\mu} + \dbar + \partial + \mu ,
            \label{eq:ac-decomp}
	\end{equation}
        where $\mu$ and its complex conjugate $\Bar\mu$ are related to the Nijenhuis tensor. No other terms are present. Defining the weighted flag structure considered above, we see that
        \deq{
            \d_1 = \Bar\mu, \quad
            \d_0 = \dbar, \quad
            \d_{-1} = \del, \quad
            \d_{-2} = \mu.
        }
	Crucially, the Dolbeault differential $\dbar$ no longer squares to zero, such that standard Dolbeault cohomology is no longer well defined. 
But we can nevertheless construct $W^\bu$ by 
first passing to the cohomology of~$\bar{\mu}$:
\begin{equation}
	W^{\bu} = H^\bu(\Omega^{\bu}(X) , \bar{\mu}) .
\end{equation}
This reproduces the construction of the Dolbeault cohomology of an almost complex manifold, as defined in~\cite{ACDolbeault}. Homotopy transfer as $D_\infty$ algebras then produces a $D_\infty$ structure on~$W^\bu$, which plays the role of the Hodge-to-de-Rham spectral sequence in this case. 

We note that the first term in the differential, $\d_1 = \Bar\mu$, can be thought of as encoding the failure of the corresponding flag of distributions  to  be  integrable. (In the theory of filtered structures, one would say that the symbol of the flag of distributions fails to be abelian.) This is further  illustrated by the next examples.

\numpar[sec:super][Superspaces and the canonical supermultiplet]
Let $\fn$ be a supertranslation algebra in the sense of~\cite{perspectives}: a super Lie algebra with a consistent weight grading supported in degrees one and two. ``Consistent'' means that parity equals weight modulo two: thus $\fn = \fn_1 \oplus \fn_2$, where $\fn_1$ has weight one and odd internal parity and $\fn_2$ has weight two and even internal parity. Let $N = \exp(\fn)$ be the corresponding flat superspace. 
The de Rham complex
\begin{equation}
	\left( \Omega^{\bu}(N) , \d_{\dR} \right) = \left( C^\infty(N_+)[\theta, \d\theta, \d x] \:  ,  \: \d x \frac{\partial}{\partial x}+ \d\theta \frac{\partial}{\partial \theta} \right)
\end{equation}
is then a cdga equipped with a weight grading.\footnote{We will sometimes think in terms of the \emph{totalized grading} when discussing $\fn$; when doing this, $\fn$ is a \emph{cohomologically} graded Lie algebra in degrees one and two. But we will emphasize such usage wherever it appears.}

We can define a flag of distributions in $TN$ by choosing $D_1$ to be spanned by the odd left-invariant vector fields $\left(\Vect(N)^N\right)_-$, and $D_2$ to be just $TN$. In physical examples in three or more dimensions, $D_1$ is always bracket-generating, since every translation is the square of some supercharge. Thus the distribution we consider is maximally noninvolutive.

This flag of distributions defines a weighted flag structure on~$\Omega^\bu(N)$. Concretely, we can express the de Rham complex in a left-invariant basis
\begin{equation}
	\lambda = \d \theta, \qquad
	v = \d x + \lambda \theta.
\end{equation}
$\lambda$ carries weight $-1$ and $v$ weight $-2$ (just as in~\S\ref{p:weighting}).  The de Rham differential is
\begin{equation}
\d_{\dR} = \lambda^2 \frac{\partial}{\partial v} + \lambda \left( \frac{\partial}{\partial \theta}  -  \theta \frac{\partial}{\partial x}\right) + v \frac{\partial}{\partial x} .
\end{equation}
Note that we suppress the contractions in the notation when there is no ambiguity.\footnote{For example, in terms of the structure constants $f^\mu_{\alpha \beta}$ of $\fn$, we have $\lambda^2 = \lambda^\alpha f^\mu_{\alpha \beta} \lambda^\beta$ and $\lambda \theta = \lambda^\alpha f^\mu_{\alpha \beta} \theta^\beta$. Here $\alpha, \beta$ label a basis of $\fn_1$ and $\mu$ a basis of $\fn_2$.}
The totalized grading on the de Rham complex is just given by 
the polynomial degree in $v$. The differential decomposes by weight into the terms
\begin{equation} \label{eq: ps-diff-split}
\begin{split}
\d_{1} & = \lambda^2 \frac{\partial}{\partial v}, \\
\d_0 &= \lambda \frac{\partial}{\partial \theta} - \lambda \theta \frac{\partial}{\partial x}, \\
\d_{-1} &= v \frac{\partial}{\partial x} \: .
\end{split}
\end{equation}
As we will see explicitly in~\S\ref{sec: pure spinors}, the generalized Dolbeault complex
%\begin{equation}
$W^\bullet$
% = \left( H^\bullet(\Omega^\bullet(N) , \d_{1})  , \: \d_0' \right)
%\end{equation}
has a natural interpretation within the pure spinor superfield formalism. In particular, the degree zero piece $W^{0,\bu}$ coincides with the canonical multiplet
of $\fn$~\cite{MSJI}; the analogue of the Dolbeault resolution of holomorphic $p$-forms is given by the multiplet associated to the $(-p)$-th Lie algebra cohomology of the supertranslation algebra $\fn$, with respect to the totalized degree. These multiplets were discussed in detail in physical examples in~\cite{perspectives}; the acyclic deformation of the differential arising from the strictly negative terms in~$\d$ was defined, and worked out concretely in examples, in~\cite{EHSequiv}.

\numpar[sec:Tanaka][Further examples; (flat) distributions of constant symbol]
In the previous sections, we have already gone through the examples that will interest us in detail in the remainder of the paper. Our main aim here is to set up the analogy between almost complex geometry and superspace by viewing them both as weighted flag structures, and to exploit this to give a geometric construction of interacting eleven-dimensional supergravity and its twists. However, numerous other structures could be viewed through this lens, and we feel it would be profitable to do so. We give a partial list of such examples, to which we hope to return in future work.
\begin{itemize}
\item[---] Any contact manifold has a flag structure.
\item[---] Any manifold equipped with a Tanaka structure~\cite[Definition 1]{TanakaAlt} has a flag structure on its de Rham complex.
\item[---] Let $\fn$ be a super Lie algebra equipped with a positive weight grading. Following~\cite{Zelenko}, we can consider the flat Tanaka structure with constant symbol $\fn$. By definition, this is the simply connected super Lie group $N = \exp(\fn)$, equipped with the flag of distributions spanned by the left-invariant vector fields in~$\fn_{\leq j}$. We observe that flat superspace is a particular example of such a flat Tanaka structure, with symbol the supertranslation algebra. It should be possible to consider non-strict examples (super $L_\infty$ algebras with positive weight gradings), using results of Getzler~\cite{Getzler}.
\item[---] Any Lie algebra equipped with a finite-length positive filtration gives rise to a flag structure on its Chevalley--Eilenberg cochains.
\item[---] Any filtered Lie algebroid gives rise to a flag structure on its Lie algebroid cochains. This is a clear generalization, both of the previous example and of a flag of distributions in the tangent bundle of a manifold. 
It should be possible to extend this definition to Courant algebroids, following~\cite{RoytenbergThesis}, and then to understand potential connections to exceptional generalized geometry. In particular, connections of Tanaka prolongation to tensor hierarchy algebras~\cite{JakobTHA} should be interesting to explore.
\end{itemize}

\section{Poisson--Chern--Simons theories via derived brackets}
\label{sec:three}

\subsection{Holomorphic Poisson--Chern--Simons theory} \label{sec: pcs}
In this section, we briefly review the construction of the standard Poisson--Chern--Simons theory, defined on a product of a Calabi--Yau twofold and an odd-dimensional smooth manifold. The theory is $\Z$-graded only when the smooth manifold is one-dimensional. Poisson--Chern--Simons theory was related to the maximal twist of eleven-dimensional supergravity in a particular omega background by Costello in~\cite{CostelloMtheory2}.

\numpar
Let $X$ be a Calabi--Yau twofold with holomorphic volume form $\Omega$. In complex dimension two, $\Omega$ is also a holomorphic symplectic structure. 
We denote the corresponding holomorphic Poisson bivector by $\pi = \Omega^{-1}$.

Recall from~\S\ref{sec:complex} above that the totalized grading places $\d z$ in degree $-1$ and $\d \Bar{z}$ in degree zero. Our construction above recovers the standard Dolbeault complex (equipped with a nonstandard grading): 
\deq{
W^\bu = \Omega^\bu(X), \quad
\d_0 = \dbar, \quad
\d_{-1} = \del.
}
Contracting with $\pi$ defines an isomorphism of $\Omega^{0,\bu}(X)$-modules
\begin{equation}
\pi: \left( \Omega^{2,\bullet}(X) \: , \: \bar{\partial} \right) \longrightarrow \left( \Omega^{0,\bullet}(X) \: , \: \bar{\partial} \right) ,\qquad 
\alpha \mapsto \pi \vee \alpha.
\end{equation}

\numpar
One can now use this data to equip the Dolbeault complex $\Omega^{0,\bu}(X)$ with the structure of a cyclic $L_\infty$ algebra. This can be done in two steps:
\begin{itemize}
	\item[1.] Turn $\Omega^{\bu}(X)$ into a BV algebra.
	\item[2.] Define the Poisson bracket on $\Omega^{0,\bu}(X)$ as a derived bracket of the BV bracket.
\end{itemize}
For the first step, note that the commutator $\Delta = [\pi , \partial]$ defines a second-order differential operator acting on $\Omega^{\bullet}(X)$, satisfying $\Delta^2 = 0$ and $\Delta(1) = 0$. Hence, we can define the Koszul bracket on $\Omega^{\bullet}(X)$ by
\begin{equation}
\{\alpha , \beta \} = (-1)^{|\alpha|}(\Delta(\alpha \beta) -\Delta(\alpha) \beta ) - \alpha \Delta(\beta) \: ,
\end{equation}
making $(\Omega^{\bullet}(X) , 1, \Delta ,\{-,-\})$ into a BV algebra. This construction is due to Koszul~\cite{Koszul}.

For the second step, we employ the derived bracket construction with respect to the differential $\partial$, as described by~\cite{Kosmann96}. The derived bracket is defined by
\begin{equation}
	[-,-]_{\partial} := \{\partial(-) , - \} .
\end{equation} 
Crucially, this bracket does not turn all of $\Omega^{\bu}(X)$ into a Lie algebra; only after restricting to an abelian subalgebra (with respect to the underived bracket $\{-,-\}$) does $[-,-]_\partial$ have the right symmetry properties. It is easy to check that the Dolbeault complex $\Omega^{0,\bu}(X)$ is indeed such a subalgebra; from this, it follows that
\begin{equation}
	\left( \Omega^{0,\bu}(X) \: , \: \bar{\partial} \: , \: [-,-]_{\partial}  \right) 
\end{equation}
is a dg Lie algebra.

Evaluating $[-,-]_\partial$ on $\alpha, \beta  \in \Omega^{0,\bullet}(X)$, we find
\begin{equation}
[\alpha , \beta]_\partial = \{\partial \alpha, \beta \} =  \pi(\partial \alpha \wedge \partial \beta) \: ,
\end{equation}
recovering the well known formula for the Poisson bracket. Together with the pairing induced by wedging with the holomorphic volume form $\Omega$ and integration, this makes $(\Omega^{0,\bullet}(X) , \bar{\partial}, [-,-]_\partial)$ into a cyclic $L_\infty$ algebra---indeed, a local $L_\infty$ algebra~\cite[Definition 3.1.3.1]{CG2} with a cyclic structure of degree $-2$. Tensoring with the de Rham forms of an odd-dimensional smooth manifold gives a local $L_\infty$ algebra with an odd-shifted cyclic structure on the product manifold. The corresponding $\Z/2$-graded BV theory is called holomorphic Poisson--Chern--Simons theory. Denoting the odd dimensional smooth manifold by $M$, the BV action of the theory evaluated on an element $\alpha \in \Omega^{(0,\bu)}(X) \otimes \Omega^\bu(M)$ reads
\begin{equation}
	S_{BV}(\alpha) = \int_{X \times M} \Omega \wedge \alpha \left( \frac{1}{2}(\bar{\partial}_X + \d_M) \alpha + \frac{1}{6} [ \alpha , \alpha ]_\partial \right) .
\end{equation}

\subsection{Homotopy Poisson--Chern--Simons theory} \label{sec: hpcs}
We now generalize the above setting to the context of~\S\ref{sec:two} in order to construct a ``homotopy'' version of Poisson--Chern--Simons theory. 
\numpar
Let $(\Omega^{\bu} , \d)$ be a cdga equipped with a weighted flag structure, and let $(W^\bu, \d')$ be the corresponding generalized Dolbeault complex. Let us assume that, with respect to the totalized grading, $W^{\bu}$ is concentrated in degrees $0$, $-1$, and $-2$. For degree reasons, the differential then splits into three pieces
\begin{equation}
	\d' = \d'_0 + \d'_{-1} + \d'_{-2} \: .
\end{equation}
Explicitly, these terms arise via homotopy transfer along the diagram~\eqref{eq: hotop data}.
\begin{equation} \label{eq: trans-diff}
\begin{split}
\d'_{0} &= i \circ \d_0 \circ p \\
\d'_{-1} &= i \circ \left( \d_0 h \d_0  + \d_{-1} \right) \circ p \\
\d'_{-2} &= i \circ \left( (\d_0 h)^2 \d_0 + \d_0 h \d_{-1} + \d_{-1} h \d_0 \right) \circ p
\end{split}
\end{equation}
Note that the square zero condition for $\d'$ implies the following identities:
\begin{equation} \label{eq: differential identities}
\begin{split}
(\d'_0)^2 &= 0 \\
[\d'_{-1}, \d'_0] &= 0 \\
(\d'_{-1})^2 + [\d'_0, \d'_{-2}] &= 0  \\
[\d'_{-1} , \d'_{-2}] &= 0 \\
(\d'_{-2})^2 &= 0.
\end{split}
\end{equation}
Here, the bracket $[-,-]$ denotes the commutator of endomorphisms. As all terms are of cohomological degree one, these are all symmetric. We further assume that there is an isomorphism 
\begin{equation}
	\pi: (W^{-2, \bu} \: , \: \d'_0) \longrightarrow (W^{0,\bu} \: , \: \d'_0 )
\end{equation}
of $W^{0,\bu}$-modules.

In summary, the $D_\infty$ structure and the pairing $\pi$ act on $W^{\bu,\bu}$ as indicated by the following diagram.
\begin{equation}
	\begin{tikzcd}
	W^{0,\bu} \arrow[rr, out= -45, in = -135 , "\d'_{-2}"] \arrow[loop , distance = 2em, "\d'_0"' ,out=65, in = 115] \arrow[r,"\d'_{-1}"] & W^{-1,\bu} \arrow[r,"\d'_{-1}"] \arrow[loop, distance = 2em, "\d'_0"' ,  out=65, in = 115] & W^{-2,\bu} \arrow[loop, distance = 2em, "\d'_0"' ,  out=65, in = 115] \arrow[ll, in= 145 , out = 35, distance = 7em, "\pi"']
	\end{tikzcd}
\end{equation}

\numpar
From this data, we now construct an $L_\infty$ structure on $W^{0,\bu}$. For this purpose we perform the appropriate generalizations of the steps described in~\S\ref{sec: pcs}. 
\begin{itemize}
	\item[1.] Turn $W^{\bu}$ into a $BV_\infty$ algebra. 
	\item[2.] Define an $L_\infty$ structure on $A^\bu = W^{0,\bu}$ using a derived bracket construction.
\end{itemize}
We will see that both steps can be viewed as instances of the derived bracket construction described by~\cite{VoronovDerived,BaVorLinf}.

\numpar
We begin by recalling the definition of a $BV_\infty$ algebra.
\begin{dfn}
	A $BV_\infty$ algebra $(A,\Delta,1)$ is a unital graded commutative algebra over $\CC$ together with a degree one linear map $\Delta : A \longrightarrow A[\![t]\!]$ which can be expanded as
	\begin{equation}
		\Delta = \frac{1}{t} \sum_{k=1}^{\infty} t^k \Delta_k ,
	\end{equation}
	such that $\Delta_k$ is a differential operator of order at most $k$ and
	\begin{equation}
		\Delta^2 = 0 \quad \text{and} \quad \Delta(1) = 0 .
	\end{equation}
\end{dfn}

One can equip both $A[\![t]\!]$ and $A$ with $L_\infty$ structures in the following way.
By identifying an element $a \in A$ by the endomorphism given by left multiplication with $a$, we can embed $A$ as an abelian subalgebra into its graded Lie algebra of endomorphisms, $(\End(A) , [-,-])$. The other way round, evaluating an endomorphism at the unit gives a right inverse to this embedding. One can define a series a series of brackets on $A[\![t]\!]$ by the following formulas~\cite{VoronovDerived}.
\begin{equation}
\{a_1, \dots , a_n \}_t = [\dots [\Delta , a_1], \dots ,a_n] (1) .
\end{equation}
This makes $A[\![t]\!]$ into an $L_\infty$ algebra. Note that the unary bracket is just given by $\Delta$, while the binary bracket is then given by the well known formula for BV algebras
\begin{equation}
\{a_1,a_2\} = \Delta(a_1 a_2) - \Delta(a_1) a_2 -(-1)^{|a_1|} a_1 \Delta(a_2) .
\end{equation}
In general, the $n$-ary bracket can be thought of as measuring the failure of the $(n-1)$-ary bracket to be a multiderivation with respect to the algebra structure.

Further, we can extract an $L_\infty$ algebra structure on $A$ by taking an appropriate limit for the parameter $t$. We define
\begin{equation}
	\{a_1, \dots , a_n\} = \lim_{t \rightarrow 0} \frac{1}{t^{n-1}} \{a_1 , \dots , a_n\}_t.
\end{equation} 
The limit makes sense because $\Delta_k$ is a differential operator of order at most $k$.
Note that, for this $L_\infty$ structure, the $n$-ary operation is generated by $\Delta_n$, i.e.
\begin{equation}
	\{a_1, \dots , a_n\} = [\dots [\Delta_n , a_1], \dots , a_n] (1) \: .
\end{equation}

\numpar
Coming back to our setting, we define the operator 
\begin{equation}
	\Delta = \Delta_1 + t\Delta_2 + t^2\Delta_3 = \d'_0 + t [\pi, \d'_{-1}] + t^2 [\pi,[\pi , \d'_{-2}]] 
\end{equation}
on $W^{\bu}[\![t]\!]$.
A direct calculation shows the following proposition. 
\begin{prop}
\label{prop:W-BV-infty}
	$(W^{\bu}, \Delta,1)$ is a $BV_\infty$ algebra. Furthermore, $W^{0,\bu}$ is an abelian subalgebra, and $W^{<0,\bu}$ is a subalgebra with respect to the bracket $\{-,-\}$.
\begin{proof}
	These statements can be shown by direct calculations. For example, we can examine $\Delta^2=0$ order by order in $t$.
	Recall the identities~\eqref{eq: differential identities} for the $D_\infty$-algebra structure on $W^{\bu}$. At order $t^0$, $\Delta^2=0$ is just the square-zero condition for $\d'_0$, while the $t^1$-term vanishes since $\d'_{-1}$ and $\d'_0$ anti-commute. For the $t^2$-piece we find
	\begin{equation}
		[\pi, \d'_{-1}]^2 + [\d'_0, \pi \d'_{-2} \pi] .
	\end{equation}
	Recall that $(\d'_{-1})^2 = -[\d'_0,\d'_{-2}]$. For degree reasons, the only term contributing to the first summand is $\pi (\d'_{-1})^2 \pi$, for which we find
	\begin{equation}
		\pi (\d'_{-1})^2 \pi = - \pi[\d'_0, \d'_{-2}] \pi = -[\d'_0 , \pi \d'_{-2} \pi] ,
	\end{equation}
	using compatibility between the the pairing and $\d'_0$. All higher order pieces vanish for degree reasons.
	The other claims are verified by similar calculations and degree arguments.
\end{proof}
\end{prop}

\numpar
Proposition~\ref{prop:W-BV-infty} sets the stage for the second step. We now apply the derived bracket construction to the differential
\begin{equation}
	\d^t = \d'_0 + t\d'_{-1} + t^2 \d'_{-2} .
\end{equation}
Again, this first endows $W^{0,\bu}[\![t]\!]$ with an $L_\infty$ structure
\begin{equation}
	\mu^t_n(a_1, \dots ,a_n) = \{ \dots \{\d^t ,a_1 \}, \dots a_n\}
\end{equation}
and then finally $W^{0,\bullet}$ by taking the limit
\begin{equation}
	\mu_n = \lim_{t \rightarrow 0} \frac{1}{t^{n-1}} \mu_n^t
\end{equation}
The $L_\infty$ structure then takes the following form
\begin{equation}
	\begin{split}
	\mu_1(\alpha) &= \d'_0 \alpha \\
	\mu_2(\alpha , \beta) &= \{\d'_{-1} \alpha , \beta\} \\
	\mu_3(\alpha, \beta, \gamma) &= \{ \{\d'_{-2} \alpha , \beta \} , \gamma\} .
	\end{split}
\end{equation}

It is useful to express this $L_\infty$ structure in terms of the pairing $\pi$.
\begin{prop} \label{prop: L-inf}
	For $\alpha,\beta,\gamma \in W^{0,\bu}$ we have
	\begin{equation}
	\begin{split}
	\mu_2(\alpha, \beta) &= \pi(\d'_{-1} \alpha \cdot \d'_{-1} \beta) \\
	\mu_3(\alpha,\beta,\gamma) &= \pi(\d'_{-2} \alpha \cdot \pi(\d'_{-1} \beta \cdot \d'_{-1} \gamma)) .
	\end{split}
	\end{equation}
\end{prop}
\begin{proof}
	For $\mu_2$ we have
	\begin{equation}
		\{\d'_{-1} \alpha , \beta\} = (-1)^{|\alpha|} \left( \pi \d'_{-1}(\d'_{-1} \alpha \cdot \beta)- (\pi (\d'_{-1})^2 \alpha ) \cdot \beta \right) \: ,
	\end{equation}
	where we already used that $[\pi,\d'_{-1}] \beta = 0$ by degree reasons. Using that $\d'_{-1}$ is a derivation for the multiplication, we find the desired result.
	
	For $\mu_3$, note that
	\begin{equation}
	\begin{split}
		\{\d'_{-2} \alpha, \beta \} &= (-1)^{|\alpha|} \left( \d'_{-1} \pi( \d'_{-2} \alpha \cdot \beta) - (\d'_{-1} \pi \d'_{-2} \alpha)  \cdot \beta \right) \\
		&= (\pi \d'_{-2} \alpha) \cdot \d'_{-1} \beta  \in W^{-1,\bu} ,
	\end{split}
	\end{equation}
	where we used that $\pi$ is an isomorphism of $W^{0, \bu}$-modules in the second step. Thus, we find
	\begin{equation}
	\begin{split}
		\{ \{\d'_{-2} \alpha, \beta\} ,\gamma\} &= \{(\pi \d'_{-2} \alpha) \d'_{-1} \beta , \gamma \} \\
		&= (-1)^{|\alpha| + |\beta|} \left[ \pi \d'_{-1} \left( (\pi \d'_{-2} \alpha) (\d'_{-1} \beta) \gamma \right) - \pi \d'_{-1} \left( (\pi \d'_{-2} \alpha) (\d'_{-1} \beta) \right) \cdot \gamma \right] \\
		&= \pi \left( (\pi \d'_{-2} \alpha) \: \d'_{-1} \beta \cdot \d'_{-1} \gamma \right)
	\end{split}
	\end{equation}
	Again, using that $\pi$ is a map of $W^{0,\bu}$-modules, we find the desired result.
\end{proof}
In the examples we are interested in and which we will discuss in the following sections, $W^{0,\bu}$ is local, i.e. arising as a sheaf of $L_\infty$ algebras on some manifold, and equipped with a pairing making it a cyclic $L_\infty$ algebra. In these instances, $(W^{0,\bullet} , \d'_0 , \mu_2 , \mu_3)$ defines a perturbative interacting BV theory, perhaps after tensoring with the de Rham complex of a smooth manifold to correct for the parity of the cyclic structure. Since the $L_\infty$ structure describing the interactions is no longer strict, we refer to such a theory as a homotopy Poisson--Chern--Simons theory.

\section{Calabi--Yau twofolds from certain Gorenstein rings}
\label{sec:four}

The construction of interactions in homotopy Poisson--Chern--Simons theory can be applied to supersymmetric field theories and their twists just by working with the example of~\S\ref{sec:super}---that is, with the standard odd distribution on superspace. 
As we will show more explicitly below, this automatically places us in the context of the pure spinor superfield formalism, as presented in~\cite{NV, perspectives, EHSequiv}.
(For the broader literature on pure spinor superfield techniques in field theory, we refer to the review~\cite{Cederwall} and to references therein.)
It remains only to check which superspaces give rise to weighted flag structures satisfying the conditions of~\S\ref{sec:three}. Of the standard superspaces that appear in physics, there are precisely three examples, corresponding to eleven-dimensional minimal supersymmetry and its two distinct twists.

We begin with a very brief reminder on the pure spinor superfield formalism, as well as its relation to twisting worked out in~\cite{spinortwist}. Then we remark on the algebraic conditions required for the generalized Dolbeault complex $(W^\bu, \d'_0)$ of a superspace to have the properties of the Dolbeault complex of a Calabi--Yau twofold, and thus to give rise to a homotopy Poisson--Chern--Simons theory using the techniques of~\S\ref{sec:three}. In~\S\ref{sec:five} below, we will show that the resulting theories are eleven-dimensional supergravity and its maximal and minimal twists.

\subsection{Pure spinor superfields for twisted field theories} 
\label{sec: pure spinors}
\numpar
In the standard pure spinor superfield formalism, one begins with the data of~\S\ref{sec:super}: a weight-graded super Lie algebra of the form
\deq{
\fn = \Pi \fn_1 \oplus \fn_2,
}
which one thinks of as a supertranslation algebra. Choosing a subalgebra $\fp_0$ of the degree-zero derivations of~$\fn$ defines an extension of the form
\deq{
0 \rightarrow \fn \rightarrow \fp \rightarrow \fp_0 \rightarrow 0,
}
which plays the role of the super Poincar\'e algebra. In the standard examples, $\fp_0$ is a direct sum $\fp_0 = \so(d) \oplus \fg_R \oplus \lie{gl}(1)$ of an orthogonal Lie algebra and another Lie algebra called $R$-symmetry, together with the abelian factor that defines the weight. $\fn_1$ is required to be a spin representation of $\so(d)$, while the degree-two piece $\fn_2 = V$ is isomorphic to the vector representation.

\numpar
There is a correspondence between supertranslation algebras and generating sets of quadratic ideals in polynomial rings. Let $R = \sym^\bu(\fn_1^\vee)$ denote the ring of polynomial functions of $\fn_1$, graded by weight. For $Q\in \fn_1$, the equations $[Q,Q]=0$ define a homogeneous quadratic ideal $I \subset R$. The quotient ring $R/I$ is the ring of functions of the space (more properly, scheme) of square-zero odd elements in $\fn$, 
\begin{equation}
	Y = \Spec(R/I),
\end{equation}
which is called the \emph{nilpotence variety} of~$\fn$~\cite{NV}.
Conversely, we can produce a super Lie algebra of supertranslation type from any finite sequence of quadratic equations. Let $R = \C[\lambda_1 ,\dots , \lambda_n]$ be the polynomial ring in $n$ variables and $I$ an ideal generated by the equations,
\begin{equation} \label{eq: ideal}
I = (\lambda^\alpha f^\mu_{\alpha \beta} \lambda^\beta), \qquad \mu = 1 \dots d, \quad \alpha , \beta = 1 \dots n .
\end{equation}
We define $\fn$ to be the two-step nilpotent super Lie algebra
\begin{equation}
\fn = \Pi S (-1) \oplus V (-2) ,
\end{equation}
equipped with the indicated weight grading.
Here $S \cong \C^n$, $V \cong \C^d$, and the only non-trivial bracket is the map
\begin{equation}
[-,-] : \sym^2(S) \longrightarrow V ,
\end{equation}
generated by the equations~\eqref{eq: ideal}---in other words, with structure constants $f_{\alpha\beta}^\mu$.

\numpar
The \emph{pure spinor superfield formalism} is a functor
\begin{equation}
A^\bullet_{R/I} : \Mod_{R/I}^{\fp_0} \longrightarrow \Mult_{\fp} 
\end{equation}
from $\fp_0$-equivariant modules over the quotient ring $R/I$---in other words, from equivariant sheaves on~$Y$---to the category of $\fp$-multiplets.
%\footnote{Throughout this work, we use the notion of multiplet as defined in~\cite{perspectives}; in short 
A $\fp$-multiplet is a cochain complex of super vector bundles, equivariant for the natural action of $\fp_+$ by affine transformations, and equipped with a homotopy action of the super Poincaré algebra $\fp$ that extends this equivariance. (More details are given in~\cite{perspectives}.)

 %One quick way %(for details see~\cite{EHSequiv}) 
 %to understand the action of 
 The functor $A^\bu_{R/I}$ is easy to understand. One recalls that the smooth functions $C^\infty(N)$ on superspace admit two commuting actions of~$\fn$, by the left- and right-invariant vector fields $Q_a$ and~$D_a$, where the index $a$ refers to a basis in~$\fn_1$. Then one constructs the cdga
 \begin{equation} \label{eq: canonical}
	A^\bu := \left( C^\infty(N) \otimes_\C R/I \: , \: \lambda^a D_a \right) ,
\end{equation}
graded by placing the generators of~$R/I$ in cohomological degree one.
In coordinates, the right-invariant vector fields take the explicit form
\deq{
D = \frac{\partial}{\partial \theta} - \theta \frac{\partial}{\partial x},
}
with the structure constants of~$\fn$ appearing in the second term. This cdga acquires the structure of a multiplet through the left-invariant fields $Q_a$; it is also an algebra over~$R/I$. 
The functor $A^\bu_{R/I}$ is just the tensor product over~$R/I$ with $A^\bu$:
\begin{equation}
	A^\bu_{R/I}(\Gamma) := A^\bu \otimes_{R/I} \Gamma ,
\end{equation} 
where $\Gamma$ is any $\fp_0$-equivariant $R/I$-module. 
$A^\bu$ was called the ``canonical multiplet'' in~\cite{MSJI} and ``the tautological filtered cdgsa'' in~\cite{spinortwist}; in the prior literature, it is often just called the (scalar) pure spinor superfield. We note that $A^\bu = A^\bu_{R/I}(R/I)$.

While this functor can be used to produce many multiplets of physical interest and freely resolve them over $C^\infty(N)$, it is not an equivalence of categories. 
As explained in~\cite{EHSequiv}, this can be remedied by extending the formalism using ideas from derived geometry. 
The extension replaces $Y$ by the cdga $C^\bu(\fn)$ placing $A^\bullet_{R/I}$ in the following diagram.
\begin{equation}
\begin{tikzcd}
\Mod_{R/I}^{\fp_0} \arrow[d] \arrow[r, "A^\bullet_{R/I}"] & \Mult_{\fp} \arrow[dl, shift left, "C^\bullet"] \\
\Mod_{C^\bullet(\fn)}^{\fp_0} \arrow[ru, shift left, "\hat{A}^\bullet"]
\end{tikzcd}
\end{equation}
Here $C^\bu = C^\bu(\fn,-)$ denotes the functor of Chevalley--Eilenberg cochains, which models the derived $\fn$-invariants of a multiplet. The functors $\hat{A}^\bu$ and $C^\bu$ establish an equivalence of dg-categories between $C^\bu(\fn)$-modules and $\fp$-multiplets.

The derived uplift is a version of Koszul duality, which is easiest to see by thinking about the sheaf $\Omega^\bu$ of de Rham forms on superspace. The cdga of global sections is acyclic, and the left-invariant de Rham forms (which are isomorphic to the cochains of~$\fn$ with trivial coefficients) clearly map to it:
\deq{
C^\bu(\fn) \cong \Omega^\bu(N)^N \subset \Omega^\bu(N).
}
Furthermore, $\fn$ acts on $\Omega^\bu(N)$ by left-invariant vector fields, preserving this subalgebra. As such, $\Omega^\bu(N)$ is a resolution of the ground field in $(C^\bu(\fn), U(\fn))$-bimodules. 
Tensoring on either side gives a correspondence between the two module categories which witnesses the equivalence.

\numpar
The explicit connection between the pure spinor superfield formalism and our discussion above is now clear. Applying the pure spinor superfield functor to $C^\bu(\fn)$ itself recovers $\Omega^\bu(N)$, expressed in the left-invariant frame discussed in~\S\ref{sec:super} (see~\cite[Lemma 3.8]{EHSequiv} for the proof):
\begin{equation}
\hat{A}^\bullet(C^\bullet(\fn)) \cong (\Omega^\bullet(N) , \d_{\dR}).
\end{equation}
Recall that the differential splits according to~\eqref{eq: ps-diff-split}; the internal differential $\d_1$ now coincides with the Chevalley--Eilenberg differential $\d_{CE}$ on $C^\bu(\fn)$. Taking cohomology with respect to $\d_{1}$, we thus recover that
\begin{equation}
	\Gr F^+_\bu W^\bu = A_{R/I}^\bu(H^\bu(\fn))
\end{equation}
via an isomorphism of sheaves of cdga's.
The differential $\d_0$ is the standard pure spinor superfield differential, so that the generalized Dolbeault complex in totalized degree $-k$---the analogue of the holomorphic $k$-forms---consists of the supermultiplet associated by $A^\bu_{R/I}$ to the Lie algebra cohomology group $H^{-k}(\fn)$, again in the totalized grading. In particular, we recover the canonical multiplet as described in~\eqref{eq: canonical} in degree zero,
\begin{equation}
	(W^{0,\bu},\d_0') = A^\bu_{R/I}(R/I),
\end{equation}
such that \emph{the canonical multiplet is identified with the structure sheaf of superspace.}
This justifies our notation $A^\bu = A^\bu_{R/I}(R/I) = (W^{0,\bu},\d_0')$ from above.

\numpar
The nilpotence variety also classifies the possible twists of theories with $\fp$ supersymmetry; these are obtained by taking invariants of a square-zero odd symmetry.
Fixing such an element $Q \in Y$, we can twist the algebra itself by defining a dg Lie algebra $\left( \fp \: , \: [Q,-] \right)$. Its cohomology $\fp_Q = H^\bu(\fp , [Q,-])$ is again a graded Lie algebra in degrees zero to two and should be viewed as the residual symmetry algebra of any theory twisted by $Q$; we denote its nilpotence variety (which encodes the possible further twists of the $Q$-twisted theory) by $Y_Q$.

We call the positively graded piece of the cohomology
\begin{equation}
	\fn_Q = H^{>0}(\fp, [Q,-])
\end{equation}
the twisted supertranslation algebra. Sometimes it is convenient to work with a quasi-isomorphic dg model for $\fn_Q$ which keeps all the even translations in degree two. To this end we define a dg Lie algebra $\tilde{\fn}_Q$ by throwing away the degree zero piece of $(\fp, [Q,-])$ while simultaneously replacing its degree one piece by the cokernel of the adjoint action of $Q$,
\begin{equation}
	\tilde{\fn}_Q = \left( \fp^1 / \Im([Q,-]) (-1) \oplus \fp^2(-2) \: , \: [Q,-] \right).
\end{equation}

Importantly, we can apply the pure spinor superfield formalism not only to $\fp$, but also to any of its twists. The results of~\cite{spinortwist} indicate that the twisting procedure commutes with the pure spinor superfield construction. In concrete terms, this means that
\begin{equation}
A^\bu(\cO_Y)^Q \cong A^\bu(\cO_{Y_Q}) \:,
\end{equation}
such that the twist of the canonical multiplet is isomorphic to the canonical multiplet of the twisted algebra.

This means that the operation of twisting is, in a sense, fully internal to the superspace: any construction which relies only on the ``(almost) complex geometry'' of the weighted flag structure of a superspace, as encoded in its generalized Dolbeault complex, should behave in \emph{exactly the same way} in any twist. (Recall, for example, that the full Dolbeault complex can be reconstructed algebraically from $\Omega^{0,\bu}$ by considering the module of K\"ahler differentials. We can thus think of the acyclic $D_\infty$ structure we construct on~$W^\bu$ as related, at least loosely, to the algebraic de Rham cohomology of the dg space $\Spec A^\bu$.)

Any theory admits a maximal twist in which $A^\bu$ reduces to a mixed Dolbeault--de~Rham complex of the standard type.\footnote{Maximal twists are characterized by being smooth points of~$Y$; the deformation problem governed by $\fn_Q$ is trivial, and no further twists are possible. Maximal twists are not unique; we are here interested in the maximal twist corresponding to the highest-dimensional smooth orbit in~$Y$.} If $n$ is the number of surviving translations in this maximal twist for a theory in $d$ dimensions, then 
\deq{
A^\bu(\O_Y)^Q = \Omega^{0,\bu}(\C^n) \otimes \Omega^\bu(\R^{d-2n}).
}
In light of the above considerations, we can bootstrap information about this maximal twist: if we have a description of an interacting theory that uses only information about the complex geometry of~$\C^n$ (or, more precisely, the THF structure on~$\C^n \times \R^{d-2n}$), \emph{then exactly the same construction should give a pure spinor model for the untwisted interacting theory---or for any other twist---when applied to the corresponding generalized Dolbeault complex.}

\subsection{The defect, the effective dimension, and the maximal twist}
To apply the construction of~\S\ref{sec: hpcs} in the pure spinor superfield formalism, we thus need to specify conditions that guarantee the existence and appropriate properties of the pairing $\pi$. In particular, we would like the generalized Dolbeault complex $W^\bu$ to exhibit the properties of the Dolbeault complex of a Calabi--Yau twofold.

\numpar
Let us fix a supertranslation algebra $\fn$ with corresponding polynomial ring $R=\sym^\bu(\fn_1^\vee)$ together with $\dim(\fn_2)$ generators for the quadratic ideal $I$ and nilpotence variety $Y$. We call the number
\begin{equation}
\mathrm{def}(\fn) = \dim(Y) - (\dim(\fn_1) -\dim(\fn_2)) = \dim(\fn_2) - \codim(Y)
\end{equation}
the \emph{defect} of $\fn$.\footnote{If $Y$ is not equidimensional, we take $\dim(Y)$ to denote the maximum of the dimensions of its irreducible pieces.} Roughly, it measures how far the generators of the ideal $I$ are from forming a regular sequence.
The following proposition shows that the defect governs the support of the Chevalley--Eilenberg cohomology of $\fn$, and thus the ``complex dimension'' of $\Spec A^\bu$.
\begin{prop}
	Let $R/I$ be a Cohen--Macaulay ring. Then, $\mathrm{def}(\fn)$ is the smallest non-negative number such that $H^{-i} (\fn)\neq 0$ for all $i \geq \mathrm{def}(\fn)$.
\end{prop}
\begin{proof}
	Recall that the Chevalley--Eilenberg complex of $\fn$ is the Koszul complex on our set of generators for the ideal $I$. Let $H^{-n}(\fn)$ be the top cohomology group. By depth sensitivity (see for example~\cite[Theorem 17.4]{Eisenbud}) of the Koszul complex one has
	\begin{equation}
	\mathrm{depth}(I,R) = \dim(V) - n .
	\end{equation}
	The Cohen--Macaulay condition implies $\mathrm{depth}(I,R) = \codim(Y)$ implies the claim. 
\end{proof}

\numpar
We can further define a local version of the defect for any orbit in the nilpotence variety. For $Q \in Y$ we set
\begin{equation}
\mathrm{def}(Q) = \dim(V) - \codim(P_0 \cdot Q) .
\end{equation}
The following lemma shows that the defect of $Q$ is equal to the number of surviving translations in a twist by $Q$.

\begin{lem}
	$\mathrm{def}(Q) = \dim(H^2(\fn, [Q,-]))$. 
\end{lem}
\begin{proof}
	Recall that
	\begin{equation}
	H^2(\fn, [Q,-]) \cong V/\Im([Q,-]) \: .
	\end{equation}
	The map $[Q,-]$ induces an isomorphism
	\begin{equation}
	\fn_1 / \ker([Q,-]) \longrightarrow \Im([Q,-]) \subseteq V.
	\end{equation}
	Let $P_0 \cdot Q$ denote the orbit of $Q$ inside $Y$. We use the inclusion $i: (P_0 \cdot Q) \hookrightarrow \fn_1$ to pull back the tangent space of $\fn_1$ at $Q$. The pullback splits as a direct sum:
	\begin{equation}
	i^*T_Q \fn_1 \cong T_Q(P_0 \cdot Q) \oplus N_Q (P_0 \cdot Q) .
	\end{equation}
	We can identify the tangent space with $\ker([Q,-])$ and the normal space with the quotient $\fn_1/\ker([Q,-])$. Thus, we find in particular
	\begin{equation}
	\codim(P_0 \cdot Q) = \dim(N_Q(P_0 \cdot Q)) = \dim(\fn_1/\ker([Q,-])) .
	\end{equation}
	and therefore
	\begin{equation}
		\begin{aligned}
	\mathrm{def}(Q) &= \dim(V) - \dim(\fn_1/\ker([Q,-])) \\
	&= \dim(V/\Im([Q,-])) = \dim(H^2(\fn,[Q,-])) ,
	\end{aligned}
	\end{equation}
	proving the claim.
\end{proof}
It follows from the proposition that the defect of $\fn$ is the local defect evaluated at a maximal twist lying in an orbit of maximal dimension. (Note that this is neither the maximum, nor the minimum, value of the local defect; $Y$ need not be---and often is not---equidimensional.)

\numpar[p:Gorenstein][Gorenstein rings of defect two]
Let us fix a supertranslation algebra $\fn$ of defect two such that the quotient ring $R/I$ is both Gorenstein and strongly Cohen--Macaualay.\footnote{A quotient ring $R/I$ is called strongly Cohen--Macaulay, when all Koszul homology groups (for $R/I$ viewed as an $R$-module) are Cohen--Macaulay~\cite{Golod}.} By construction, the zeroth Chevalley--Eilenberg cohomology of $\fn$ yields,
\begin{equation}
	H^0(\fn) = R/I .
\end{equation}
Further, $H^\bullet(\fn)$ is concentrated in degrees $0,-1$ and $-2$. Since $R/I$ is strongly Cohen--Macaulay, $H^\bullet(\fn)$ is a Poincar\'e duality algebra~\cite{AvramovGolod, Golod}. In particular, we have
\begin{equation}
	H^{-2}(\fn) \cong \Ext_R^{\codim(Y)} (R/I, R) \cong R/I ,
\end{equation}
where we used the Gorenstein property for the last identification.
Thus, there is an isomorphism of $A^\bu(H^0(\fn))$-modules
\begin{equation}
\pi : \left( A^\bullet(H^{-2}(\fn)) \: , \: \d'_0 \right) \longrightarrow \left( A^\bu(H^0(\fn)) \:  , \:  \d'_0 \right) .
\end{equation}
As we assumed that the defect of the supertranslation algebra equals two, transfer of the $D_\infty$ along~\eqref{eq: hotop data} yields an induced $D_\infty$ structure given by~\eqref{eq: trans-diff}.

Hence, we are in the situation described in~\S\ref{sec: hpcs} and can construct an $L_\infty$ structure on $A^\bullet(H^0(\fn))$. Furthermore, the Gorenstein property implies that there is another pairing on $A^\bullet(H^0(\fn))$, (see~\cite{perspectives}), making it a cyclic $L_\infty$ algebra and hence an interacting BV theory (after taking the product with an odd-dimensional smooth manifold to adjust the parity of the cyclic structure, if necessary).

\section{Eleven-dimensional supergravity, both twisted and not}
\label{sec:five}

As mentioned above, there are three significant examples of ``Calabi--Yau twofolds'' that arise from superspaces relevant to physics. They are all connected to eleven-dimensional supergravity: either the full theory, or one of its two twists. In this section, we review the construction of these Gorenstein rings of defect two, and then construct the corresponding homotopy Poisson--Chern--Simons theories. These recover Cederwall's pure spinor formulation of eleven-dimensional supergravity, Costello's description of the maximal twist in terms of holomorphic Poisson--Chern--Simons theory, and a pure spinor lift of the interactions of minimally twisted eleven-dimensional supergravity described in~\cite{RSW11d}. We also recall how the rings are related to one another by twists of the corresponding super Poincar\'e algebras, which shows (following~\cite{spinortwist}) that the three interacting theories are also related by twisting.

\subsection{Eleven-dimensional supersymmetry and its twists} \label{sec: 11d}
\numpar
Let $V$ denote the vector representation for $\Spin(11)$ and $S$ the unique spinor representation of dimension 32. The super Poincar\'e algebra in eleven dimensions is of the form
\begin{equation}
	\fp = \so(V) \oplus S (-1) \oplus V (-2) .
\end{equation}
The nilpotence variety $Y \subset S$ is of dimension 23, so that $\mathrm{def}(Y)=2$. Furthermore, its coordinate ring, which is the quotient of polynomial functions on~$S$ by the quadratic ideal generated by the eleven gamma matrices, is a Gorenstein ring. In this sense, the generalized Dolbeault cohomology of eleven-dimensional superspace describes a Calabi--Yau twofold.
Furthermore, the structure sheaf of this space is nothing other than the eleven-dimensional supergravity multiplet, described with a pure spinor superfield in the BV formalism~\cite{HowePS2, CederwallM5}. As was emphasized in~\cite{spinortwist,MSJI}, eleven-dimensional supergravity is a canonical supermultiplet, and is thus equipped with a \emph{commutative} structure on the space of fields.

\numpar[p:twists][Twists]
The nilpotence variety decomposes into two orbits for $\Spin(V)$, as such, there are two distinct twists available. The maximal twist is holomorphic in four directions and topological in the remaining seven, the minimal twist is holomorphic in ten directions and topological in the remaining one. The maximal twist is a smooth point of~$Y$, whereas the minimal twist corresponds to a singular point. As is well-known~\cite{BNCharacter}, the singularities take the form of the cone over the projective variety $\Gr(2,5)$. 
The stabilizer of a minimal supercharge is $\SU(5)$, whereas the stabilizer of a maximal supercharge is $G_2 \times \SU(2)$.
(For more details on the geometry, see~\cite{NV} and the references therein.)

Applying the pure spinor functor to the coordinate rings of the twisted nilpotence varieties $Y_Q$, one obtains the BV complexes of the free twisted theories. We can now apply our results to construct interactions for these theories in all these cases in a uniform way, realizing them as homotopy Poisson--Chern--Simons theories. By the results of~\cite{spinortwist}, the resulting theories are then obtained from one another---in particular, from eleven-dimensional supergravity---by taking the corresponding twist.

We will begin by describing the maximal twist, and work up to the full theory.

\subsection{The maximal twist}
In~\cite{CostelloMtheory2} a description of the maximal twist in terms of Poisson--Chern--Simons theory was proposed. The twist was computed in the free limit using a component field description in~\cite{MaxTwist}, and realized as a further twist of the minimal twist in~\cite{RSW11d}.
\numpar[p:maxstab][Decomposition under the stabilizer]
 The maximal twist on flat spacetime is defined on $\R^7 \times \CC^2$. We begin by decomposing all relevant $\Spin(11)$-representations to $G_2 \times \SU(2)\times \U(1)$. It is useful to remember the inclusions of subgroups,
\begin{equation}
	\Spin(11) \supset \Spin(7) \times \Spin(4) \supset G_2 \times \SU(2)_+ \times \SU(2)_- \supset G_2 \times \SU(2)_+ \times \U(1) ,
\end{equation}
where we identified $\Spin(4) \cong \SU(2)_+ \times \SU(2)_-$ and the $\U(1)$ appearing in the last step is the Cartan of $\SU(2)_-$. Under this subgroup, the vector representation of $\Spin(11)$ decomposes as
\begin{equation}
	V = V_7 \oplus L \oplus L^\vee ,
\end{equation}
where $V_7$ is the seven-dimensional irreducible representation of $G_2$ and $L \cong \textbf{2}^1$ (as well as $L^\vee \cong \textbf{2}^{-1}$) as $\SU(2) \times \U(1)$-representations. The spin representation gives
\begin{equation}
	S = (\textbf{1}_{G_2} \oplus V_7) \otimes (\textbf{2}^0 \oplus \textbf{1}^1 \oplus \textbf{1}^{-1}) .
\end{equation}
We immediately see that $S$ contains two copies of the trivial representation of $G_2 \times \SU(2)$, coming with $\U(1)$ weights $\pm1$. These correspond to the maximal square-zero supercharges. For definiteness, we choose
\begin{equation}
	Q \in \textbf{1}_{G_2} \otimes \textbf{1}^{-1} .
\end{equation}

\numpar[p:maxalg][Twisting the supersymmetry algebra]
Remembering that $\so(V) \cong \wedge^2 V$ and that, as $G_2$-representations,
\begin{equation}
	\wedge^2 V_7 \cong V_7 \oplus \fg_2 ,
\end{equation}
we can decompose the dg Lie algebra $(\fp, [Q,-])$ as shown in Table~\ref{tab:1}.
\begin{table}[ht]
\begin{equation}
	\begin{tikzcd}[row sep = tiny]
	V_7 \arrow[r] &	V_7 \otimes \textbf{1}^{-1} & \\
	\fg_2&	 V_7 \otimes \textbf{1}^1 \arrow[r] &		V_7\\
	V_7 \otimes \textbf{2}^1 \arrow[r] &	V_7 \otimes \textbf{2}^0 & \textbf{2}^1\\
	V_7 \otimes \textbf{2}^{-1} & \textbf{2}^0 \arrow[r] & \textbf{2}^{-1} \\
	\textbf{1}^0 \arrow[r] & \textbf{1}^{-1}\\
	\textbf{1}^{2} \arrow[r] & \textbf{1}^1 \\
	\textbf{1}^{-2} \\
	\textbf{3}^0 \\
	\end{tikzcd}
\end{equation}
\caption{Decomposition under the stabilizer}
\label{tab:1}
\end{table}
Here, the arrows represent the map $[Q,-]$. By Schur's lemma, all non-vanishing arrows are multiples of the identity; thus it is immediate to compute the cohomology. Identifying the holomorphic translations as $\textbf{2}^1 = L$, we find a purely even Lie algebra of the form
\begin{equation}
	\fp_Q = H^\bu(\fp,[Q,-]) = \left( \fg_2 \oplus \mathfrak{sl}(L) \oplus V_7 \otimes \textbf{2}^{-1} \oplus \textbf{1}^{-2} \right) \oplus  L (-2) .
\end{equation}
We note that the positively graded piece $\fn_Q$ is just the abelian even algebra $L$. The dg model $\tilde{\fn}_Q$ is of the form 
%\begin{table}
	\begin{equation} \label{eq: dg twisted transl}
	\begin{tikzcd}[row sep = tiny]
	V_7 \otimes \textbf{1}^1 \arrow[r] & V_7 \\
	\textbf{2}^0 \arrow[r] & \textbf{2}^{-1} \\
	& \textbf{2}^{1} 
	\end{tikzcd}.
	\end{equation}
%	\caption{The dg model $\tilde{\fn}_Q$.}
%\end{table}
Correspondingly, $\O_{Y_Q}= \C$, and the nilpotence variety is just a point.
Note that both the dimension as well as the codimension are zero. As there are two surviving translations, the defect is thus $\mathrm{def}(\fn_Q) = 2$.

\numpar
We can now apply the formalism of~\S\ref{sec:super} to the twisted supertranslation algebra $\fn_Q$. The weighted flag structure takes $D_1$ to be the zero section and $D_2$ to be the full (holomorphic) tangent bundle.
Doing so, we recover the negatively graded algebraic de Rham complex of~$\C^2$:
\deq{
\Omega^\bu = \C[z_1,z_2][\d z_1,\d z_2],
}
with $\d z_i$ in totalized degree $-1$. The differential $\d_1$ is trivial, and $W^\bu = \Omega^\bu$; the ``structure sheaf,'' which is the canonical multiplet of~$\fn_Q$, just consists of holomorphic functions on~$\C^2$.

In order to give a representation as a multiplet living on $V = \R^7 \times \C^2$, we can resolve in smooth functions over $V$; this recovers the Dolbeault complex of $(0,\bu)$ forms on~$\C^2$. We note that this can be obtained directly by considering the canonical multiplet of the dg model $\fn_Q$:  
\begin{equation} \label{eq: maxtwist}
	A^{\bu}(\cO_{Y_{Q}}) \simeq \left( \Omega^{0,\bu}(\CC^2) \otimes \Omega^{\bu}(\R^7) \: , \: \bar{\partial}_{\CC^2} + \d_{\R^7} \right)
\end{equation}
In either case, this corresponds to the field content of the maximal twist of eleven-dimensional supergravity~\cite{MaxTwist}. 

We thus find ourselves in the setting of $\Z/2$-graded holomorphic Poisson--Chern--Simons theory. Constructing the $L_\infty$ structure recovers the interactions of Poisson--Chern--Simons described in~\S\ref{sec: pcs}. 

We note that the vanishing of the Chevalley--Eilenberg differential on the twisted supertranslation algebra (which directly follows from maximality of the twist) ensures that we end up with Poisson--Chern--Simons theory instead of its homotopy version. This is a general feature of maximal twists. Nonetheless, applying our construction to a non-integrable complex structure would have given rise to a non-strict Poisson--Chern--Simons theory with nonvanishing 3-ary bracket.

\subsection{The minimal twist}
The minimal twist was computed in the free limit at the pure spinor cochain level in~\cite{spinortwist}. Interactions for the component fields were proposed (and numerous consistency checks perfomed) in~\cite{RSW11d}.

\numpar[p:minstab][Twisting the supersymmetry algebra]
 The stabilizer of a minimal square-zero supercharge $Q \in Y$ is isomorphic to $\SU(5)$. Choosing such a $Q$ is equivalent to the choice of a maximal isotropic subspace $L \subset V$. The vector representation then decomposes as
\begin{equation}
	V = L \oplus L^\vee \oplus \C .
\end{equation}
The twisted super Poincar\'e algebra $(\fp, [Q,-])$ and its cohomology $\fp_Q$ were analyzed in~\cite{spinortwist}. The positively graded piece of the cohomology is found to be 
\begin{equation}
	\fn_Q \cong \Pi \wedge^2 L (-1) \oplus \wedge^4 L (-2) , 	
\end{equation}
where the bracket of two odd elements is given by the wedge product. (The parentheses refer to shifts in the weight grading.) The nilpotence variety $Y_Q$ is isomorphic to the affine cone over the the Grassmannian $\Gr(2,5)$ of two-planes inside a five-dimensional vector space. One can equivalently think of this as the space of bilinear skew forms of rank two on~$L^\vee$.
As an affine variety, we have $\dim(Y_Q) = 7$, and therefore
\begin{equation}
	\mathrm{def}(\fn_Q) = 7 - (10 - 5) = 2 .
\end{equation}

\numpar $\cO_{Y_{Q}}$ is also Gorenstein, so that we can apply our procedure to construct interactions for $A^\bu(\cO_{Y_{Q}})$.
By~\cite{spinortwist}, the pure spinor multiplet $A^\bu(\cO_{Y_Q})$ is equivalent to the minimal twist of the supergravity multiplet.
Our procedure thus constructs interactions for minimally twisted supergravity on the pure spinor cochain level, corresponding to a suggestion in~\cite{Ced-SL5}. We expect that the interacting theory with this field content constructed in~\cite{RSW11d} can be obtained from this cochain-level description via homotopy transfer, thus rigorously proving that the twisted eleven-dimensional supergravity theory of~\cite{RSW11d}---which is intimately related to the exceptional simple linearly compact super Lie algebra $E(5|10)$---is in fact the twist of eleven-dimensional supergravity. This computation will appear in forthcoming work~\cite{FHIS-transfer}.

\numpar[p:CE] From above, we know that $W^\bu$ can be constructed by considering the pure spinor multiplets associated to the Lie algebra cohomology groups of~$\fn_Q$. The cochains are given by
\begin{equation}
	C^\bu(\fn_Q) \cong \wedge^\bu L^\vee \otimes R ,
\end{equation}
where we identified 
\begin{equation}
	\sym^\bu(\fn_1^\vee) = R = \CC[\lambda^{ab}] .
\end{equation}
We think of $\lambda^{ab}$ as a basis on $(\fn_Q)_1^\vee = (\wedge^2 L)^\vee$ for $a,b=1,\dots, 5$, and make use of the isomorphism $\wedge^4 L \cong L^\vee$. Further, we can think of $L^\vee$ as constant holomorphic one-forms on $L = \CC^5$ with basis $\{\d z^a\}$. The Chevalley--Eilenberg differential is of the form
\begin{equation}
\d_{CE} = \lambda^{ab} \lambda^{cd} \varepsilon_{abcde} \frac{\partial}{\partial(\d z_e)} .
\end{equation}
As expected for a Gorenstein ring of defect two, the cohomology is concentrated in degrees $0,-1$ and $-2$:
\begin{equation}
	H^{k} (\fn_Q) \cong \begin{cases}
	R/I &k \in \{0,-2\} \\
	M  &k=-1\\
	0  &\text{else},
	\end{cases}
\end{equation}
where $M$ is the cokernel of the map
\begin{equation}
	\phi : R \otimes \wedge^2 L \longrightarrow R \otimes L \qquad e_a \wedge e_b \mapsto \varepsilon^{abcde} \lambda_{cd} e_e .
\end{equation}
Here $\{e_a\}$ is a basis of $L$. As $R/I$-modules, $H^0(\fn_Q)$ is freely generated by the unit 1, while $H^{-2}(\fn_Q)$ is freely generated by $\lambda_{ab} \d z^a \d z^b$.

\numpar
After tensoring with de Rham forms on~$\R$ in order to  resolve freely over $\CC^5 \times \R$, we can describe $\Omega^\bu$ with the quasi-isomorphic complex
\begin{equation}
	\left( \Omega^{\bu}_\text{dR}(\CC^5) \otimes \CC[\lambda^{ab}, \theta^{ab}] \: , \: \partial_{\CC^5} + \bar{\partial}_{\CC^5} + \sR + \d_{CE} \right) \otimes \left( \Omega^\bu(\R) \: , \: \d_{\R} \right) ,
\end{equation}
where 
\deq{\sR = \lambda \left( \pdv{ }{\theta} - \theta \pdv{ }{x} \right)}
is the standard pure spinor differential. Here, the spatial coordinate $x$ is one of $(z,\Bar{z})$.
Note that, with respect to the description in~\S\ref{sec:super}, $\theta$ is an odd function on the superspace $N$, whereas the one-forms are $\lambda$, $\d z$, and $\d \Bar{z}$. The weighted flag structure places $\d z$ in totalized degree $-1$ and everything else in degree zero.
We can identify
\begin{equation}
	\d_{1} = \d_{CE}, \quad \d_0 = \sR + \bar{\partial} + \d_\R, \quad \d_{-1} = \partial .
\end{equation}

We construct the generalized Dolbeault complex $W^\bu$ according to the standard procedure, using the formulas for the transferred $D_\infty$ structure above~\eqref{eq: trans-diff}. The weighted pieces of the generalized Dolbeault complex are the pure spinor multiplets associated to the modules of~\S\ref{p:CE}.
In contrast to the maximal twist, a piece of degree $-2$ arises, such that there is a non-vanishing map
\begin{equation}
	\d'_{-2} : A^\bu(H^0(\fn_Q)) \longrightarrow A^\bu(H^{-2}(\fn_Q)) ,
\end{equation}
signaling that the induced $L_\infty$ structure will not be strict.

The Gorenstein property guarantees that there is an isomorphism
\begin{equation}
	\pi : \left( W^{-2,\bu}, \d'_0 \right) \longrightarrow \left( W^{0,\bu}, \d'_0 \right) .
\end{equation}
Explicitly, $\pi$ is induced from the isomorphism between $H^{-2}(\fn_Q)$ and $H^0(\fn_Q)$; thus, in terms of representatives, we have
\begin{equation}
	\pi(\lambda_{ab} \d z^a \d z^b) = 1 .
\end{equation}
Hence, we obtain an $L_\infty$ algebra structure on $A^\bu(H^{0}(\fn_Q))$ by the formulas in Proposition~\ref{prop: L-inf}.

\subsection{Eleven-dimensional supergravity}
Recall that the canonical multiplet associated to the eleven-dimensional supertranslation algebra is the supergravity multiplet. In~\cite{Ced-towards} and~\cite{Ced-11d}, Cederwall constructed a consistent quartic BV action functional, recovering interacting eleven-dimensional supergravity in the pure spinor superfield formalism. We now recover these interactions as an instance of homotopy Poisson--Chern--Simons theory.

\numpar[p:untwist][Lie algebra cohomology and $W^\bu$] The Chevalley--Eilenberg cochains of the untwisted supertranslation algebra take the form
\begin{equation}
	C^\bu(\fn) = \left( \wedge^\bu V^\vee \otimes R \: ,\: \d_{CE} \right) ,
\end{equation}
where $R= \sym^\bu(S^\vee) = \CC[\lambda^\alpha]$ is the polynomial ring in $\{\lambda^\alpha\}$ with $\alpha= 1,\dots,32$. Fixing a basis $\{v_\mu\}$ of $V^\vee$, the Chevalley--Eilenberg differential takes the form
\begin{equation}
	\d_{CE} = \lambda^\alpha \Gamma^\mu_{\alpha \beta} \lambda^\beta \frac{\partial}{\partial v^\mu} .
\end{equation}
Again, Chevalley--Eilenberg cohomology is concentrated in degrees $0,-1$ and $-2$, with $H^0(\fn)$ and $H^{-2}(\fn)$ both being isomorphic to the ring of functions on the nilpotence variety $\cO_Y = R/I$. The cohomology in degree $-2$ is spanned by the class
\begin{equation}
	(\lambda^\alpha \Gamma^{\mu \nu}_{\alpha \beta} \lambda^\beta ) v_\mu v_\nu .
\end{equation}

\numpar[p:11dPoisson][Eleven-dimensional interactions] Applying the pure spinor superfield construction, we construct the generalized Dolbeault complex as the sum of the multiplets associated to the modules from the previous section. (We note that the multiplet $W^{-1,\bu}$ physically corresponds to a field-strength multiplet for~$W^{0,\bu}$; this fact was already appreciated in~\cite{CederwallM5}.)

As always, the weighted flag structure on the de Rham complex of superspace induces a $D_\infty$-module structure on~$W^\bu$, where $\d'_0$ is the standard pure spinor differential and $\d'_{-1}$ and $\d'_{-2}$ are both nontrivial.
Restricting these differentials to $W^{0,\bu}$ 
recovers Cederwall's differential operators constructed in~\cite{Ced-towards,Ced-11d}, where
\begin{equation}
	\d'_{-1} : W^{0,\bu} \longrightarrow W^{-1,\bu} ,
\end{equation}
corresponds to ``$R$'' and 
\begin{equation}
	\d'_{-2} : W^{0,\bu} \longrightarrow W^{-2,\bu}
\end{equation}
corresponds to ``$T$''. Together with $\pi$ induced from
\begin{equation}
	\pi(\lambda^\alpha \Gamma^{\mu \nu}_{\alpha \beta} \lambda^\beta v_\mu v_\nu ) = 1 ,
\end{equation}
this yields an $L_\infty$-structure on $W^{0,\bu}$.

\numpar[integration][Batalin--Vilkovisky actions?] We have constructed the $L_\infty$ structure underlying the interactions of eleven-dimensional supergravity as a homotopy Poisson--Chern--Simons theory, working at the pure spinor cochain level. It is also known that the complexes we work with admit $(-1)$-shifted (or at least odd-shifted) symplectic structures. Nevertheless, writing BV action functionals would require a theory of integration or Verdier duality on the dg ringed spaces we construct. This has not yet been worked out concretely in our examples, and so we refrain from writing such functionals here (although we note that it can be done using an analytic approach to constructing singular Calabi--Yau forms on the nilpotence variety; see~\cite{Ced-11d}, where the action functional in the untwisted case is written). We look forward to studying the theory of integration on pure spinor superspace in future work.

\printbibliography
\end{document}